\documentclass[11pt]{article}
\usepackage{graphicx,xspace,amsmath,amssymb,amsthm,subfig,fullpage}
\usepackage[colorlinks=true,pdfpagemode=UseNone,urlcolor=blue,linkcolor=blue,citecolor=BrickRed,pdfstartview=FitH]{hyperref}
\graphicspath{{./figures/}}

\usepackage[usenames,dvipsnames]{color}

\newcommand{\cclass}[1]{\ensuremath{\mbox{\textup{#1}}}\xspace}

\newcommand{\NL}{\cclass{NL}}
\newcommand{\NP}{\cclass{NP}}
\newcommand{\NLk}{\cclass{EPNL}}
\newcommand{\SNP}{\cclass{SNP}}

\newcommand{\ceil}[1]{\lceil #1 \rceil}

\newcommand{\nlkreduce}{\leq_L^{pp}}

\newcommand{\tw}{{\bf tw}\xspace}
\newcommand{\cw}{{\bf cw}\xspace}

\newcommand{\pw}{{\bf pw}\xspace}

\newcommand{\C}{{\mathcal C}\xspace}
\newcommand{\poly}{\mathrm{poly}\xspace}
\newcommand{\true}{\mathrm{true}\xspace}

\newcommand{\problem}[1]{\textsc{#1}\xspace}
\newcommand{\SAT}{\problem{SAT}}
\newcommand{\ThreeSAT}{\problem{3-SAT}}
\newcommand{\MaxTwoSAT}{\problem{Max 2-SAT}}

\newcommand{\IS}{\problem{Independent Set}}
\newcommand{\DS}{\problem{Dominating Set}}

\newcommand{\HIT}{\problem{Hitting Set}}
\newcommand{\SPLIT}{\problem{Set Splitting}}
\newcommand{\SetCover}{\problem{Set Cover}}
\newcommand{\TSP}{\problem{Directed Hamiltonicity}}

\newtheorem{theorem}{Theorem}
\newtheorem{lemma}{Lemma}

\newtheorem{definition}{Definition}
\newtheorem{proposition}{Proposition}

\begin{document}

%\mainmatter

\title{On the Equivalence among Problems of Bounded Width}

\author{Yoichi Iwata
	\thanks{The University of Tokyo
		\texttt{y.iwata@is.s.u-tokyo.ac.jp}
	}\and
  	Yuichi Yoshida
  	\thanks{National Institute of Informatics and Preferred Infrastructure, Inc.
  	\texttt{yyoshida@nii.ac.jp} }
}

\maketitle

\begin{abstract}
  In this paper, we introduce a methodology, called decomposition-based reductions, for showing the equivalence
  among various problems of bounded-width.

  First, we show that the following are equivalent for any $\alpha > 0$:
  \begin{itemize}
  \itemsep=0pt
  \item \SAT can be solved in $O^*(2^{\alpha\tw})$ time,
  \item \ThreeSAT can be solved in $O^*(2^{\alpha\tw})$ time,
  \item \MaxTwoSAT can be solved in $O^*(2^{\alpha\tw})$ time,
  \item \IS can be solved in $O^*(2^{\alpha\tw})$ time, and
  \item \IS can be solved in $O^*(2^{\alpha\cw})$ time,
  \end{itemize}
  where $\tw$ and $\cw$ are the tree-width and clique-width of the instance, respectively.
%  It is surprising that such connections among these problems exist since the known fastest algorithm for \IS of
  % bounded clique-width requires fast subset convolution whereas the others do not, and \IS of bounded clique-width seems a harder problem.

  Then, we introduce a new parameterized complexity class $\NLk$, which includes \SetCover and \TSP, and show that \SAT,
  \ThreeSAT, \MaxTwoSAT, and \IS parameterized by path-width are $\NLk$-complete.
  This implies that if one of these $\NLk$-complete problems can be solved in $O^*(c^k)$ time, then any problem in
  $\NLk$ can be solved in $O^*(c^k)$ time.

\end{abstract}

%!TEX root = ../paper.tex

\section{Introduction}\label{sec:intro}
\SAT is a fundamental problem in complexity theory.
Today, it is widely believed that \SAT cannot be solved in polynomial time.
This is not only because anyone could not find a polynomial-time algorithm for \SAT despite many attempts, but also
because if \SAT can be solved in polynomial time, any problem in \NP can be solved in polynomial time
(\NP-completeness).
Actually, even no algorithms faster than the trivial
$O^*(2^n)$-time\footnote{$O^*(\cdot)$ hides a factor polynomial in the input size.} exhaustive search algorithm are
known.
Impagliazzo and Paturi~\cite{DBLP:journals/jcss/ImpagliazzoP01} conjectured that \SAT cannot be solved in
$O^*((2-\epsilon)^n)$ time for any $\epsilon > 0$, and this conjecture is called the \emph{Strong Exponential Time Hypothesis (SETH)}.
Under the SETH, conditional lower bounds for several problems have been obtained, including
\problem{$k$-Dominating Set}~\cite{DBLP:conf/soda/PatrascuW10}, problems of bounded
tree-width~\cite{DBLP:conf/soda/LokshtanovMS11a,DBLP:conf/focs/CyganNPPRW11}, and \problem{Edit
Distance}~\cite{DBLP:conf/stoc/BackursI15}.

When considering polynomial-time tractability, all the \NP-complete problems are equivalent,
that is, if one of them can be solved in polynomial time, then all
of them can be also solved in polynomial time.
Similarly, when considering subexponential-time tractability, all the \SNP-complete problems are
equivalent~\cite{DBLP:journals/jcss/ImpagliazzoPZ01}.
However, if we look at the exponential time complexity for solving each \NP-complete problem more closely, the situation
changes; whereas the current fastest algorithm for \SAT is the naive
$O^*(2^n)$-time exhaustive search algorithm,
faster algorithms have been proposed for many other NP-complete problems such as
\ThreeSAT~\cite{DBLP:journals/siamcomp/Hertli14}, \MaxTwoSAT~\cite{DBLP:journals/tcs/Williams05}, and
\IS~\cite{DBLP:conf/isaac/XiaoN13}.
Although there are many problems, including \SetCover and \TSP\footnote{For \problem{Undirected Hamiltonicity}, a
faster algorithm has been proposed in a recent paper by Bj\"orklund~\cite{DBLP:journals/siamcomp/Bjorklund14}. However, for \TSP, the trivial
$O^*(2^n)$-time dynamic programming algorithm is still the current fastest.}, for which the current fastest algorithms
take $O^*(2^n)$ time, we do not know whether a faster algorithm for one of these problems leads to a faster algorithms for \SAT and vice versa.
Actually, only a few problems, such as \HIT and \SPLIT, are known to be equivalent to \SAT
in terms of exponential time complexity~\cite{DBLP:conf/coco/CyganDLMNOPSW12}.

In this paper, we propose a new methodology, called \emph{decomposition-based reductions}.
Although the idea of decomposition-based reductions is simple, we can obtain various interesting results.
First, we show that when parameterized by \emph{width}, there are many problems that are equivalent to \SAT.
Second, we show the equivalence among different width; \IS parameterized by tree-width and \IS parameterized
by clique-width are equivalent.
Third, we introduce a new parameterized complexity class \emph{$\NLk$}, which includes \SetCover and \TSP, and
show that many problems parameterized by path-width are $\NLk$-complete.
For these problems, conditional lower-bounds under the SETH are already known~\cite{DBLP:conf/soda/LokshtanovMS11a}.
However, our results imply that these problems are at least as hard
as not only $n$-variable \SAT but also \emph{any} problem in $\NLk$.
In this sense, our hardness results are more robust.

It has been shown that many NP-hard graph optimization problems can be solved
efficiently if the input graph has a nice decomposition.
One of the most famous decompositions is \emph{tree-decomposition},
and a graph is parameterized by \emph{tree-width},
the size of the largest bag in the (best) tree-decomposition of the graph.
Intuitively speaking, tree-width measures how much a graph looks like a tree.
If we are given a graph and its tree-decomposition of width \tw\footnote{Obtaining a tree-decomposition of the minimum
width is \NP-hard. In this paper, we assume that we are given a decomposition as a part of the input, and a problem
is parameterized by the width of the given decomposition.}, many problems can be solved in $O^*(c^\tw)$ time, where $c$
is a problem-dependent constant.
For example, we can solve \IS and \MaxTwoSAT in $O^*(2^\tw)$ time by standard dynamic programming and \DS in
$O^*(3^\tw)$ time by combining with subset convolution~\cite{DBLP:conf/esa/RooijBR09}.\footnote{For problems
related to SAT, we consider the tree-width of the primal graph of the input. See Section~\ref{sec:pre} for details.}
Recently, Lokshtanov~\emph{et~al.}~\cite{DBLP:conf/soda/LokshtanovMS11a} showed that many of these algorithms are
optimal under the SETH.
These results are obtained by reducing an $n$-variable instance of \SAT to an instance of the target problem with tree-width approximately $\frac{n}{\log c}$, where $c$ is a problem dependent constant.
However, these reductions are one-way, and thus a faster \SAT algorithm may not lead to faster algorithms for
these problems.
Moreover, there is a possibility that one of these problems has a faster algorithm but the others do not.

The first contribution of this paper is showing the following equivalence among  problems of bounded tree-width:
\begin{theorem}\label{thm:tw:eq}
  For any $\alpha>0$, the following are equivalent:
  \begin{enumerate}
    \item \SAT can be solved in $O^*(2^{\alpha\tw})$ time.
    \item \ThreeSAT can be solved in $O^*(2^{\alpha\tw})$ time.
    \item \MaxTwoSAT can be solved in $O^*(2^{\alpha\tw})$ time.
    \item \IS can be solved in $O^*(2^{\alpha\tw})$ time.
  \end{enumerate}
\end{theorem}
For all of these problems, the fastest known algorithms run in $O^*(2^\tw)$
time~\cite{Niedermeier2006} and Theorem~\ref{thm:tw:eq} states that this is not a coincidence.
Note that an $n$-variable instance of \SAT has tree-width at most $n-1$.
Hence by Theorem~\ref{thm:tw:eq}, for any $\epsilon > 0$,
an $O^*((2-\epsilon)^{\tw})$-time algorithm for \IS of bounded tree-width implies an $O^*((2-\epsilon)^{n})$-time
algorithm for the general \SAT.
Therefore, our result includes the hardness result by Lokshtanov~\emph{et~al.}~\cite{DBLP:conf/soda/LokshtanovMS11a}.
We believe that the same technique can be applied to many other problems.
In practice, SAT solvers are widely used to solve various problems by reductions to \SAT.
Using our methodology, we can reduce an instance of some problem to an instance of \SAT by preserving the tree-width.
Since tree-decompositions can be used to speed-up SAT solvers~\cite{DBLP:conf/ictai/HabetPT09}, our reductions may be
useful in practice.

\emph{Clique-width} is the number of labels we need to construct the given graph by iteratively performing
certain operations.
% (see Appendix~\ref{sec:is-cw-to-sat-tw} for details).
Similarly to the tree-width case, many problems can be solved in $O^*(c^\cw)$ time if the given graph has a clique-width $\cw$, where $c$ is a problem-dependent constant~\cite{Courcelle:2000jv}.
%Though any graph of a tree-width $\tw$ has a clique-width $\cw = O(2^\tw)$~\cite{Corneil:2006ic}, the other such direction cannot hold.

The second contribution of this paper is showing the following equivalence between \IS of bounded tree-width and bounded
clique-width:
\begin{theorem}\label{thm:cw:eq}
  For any $\alpha>0$, the following are equivalent:
  \begin{enumerate}
    \item \IS can be solved in $O^*(2^{\alpha\tw})$ time.
    \item \IS can be solved in $O^*(2^{\alpha\cw})$ time.
  \end{enumerate}
\end{theorem}
The fastest known algorithms for \IS parameterized by clique-width runs in $O^*(2^\cw)$
time~\cite{Courcelle:2000jv}.
It is surprising that we can obtain such strong connections between problems of bounded tree-width and a
problem of bounded clique-width because tree-width and clique-width are very different parameters in nature; a complete
graph of $n$ vertices has a clique-width two whereas its tree-width is $n-1$.
Hence, even if there is an efficient algorithm for a problem of bounded tree-width, it does not immediately imply that there is an efficient algorithm for the same problem of bounded clique-width.
However, Theorem~\ref{thm:cw:eq} states that a faster algorithm for \IS of bounded tree-width implies a faster algorithm
for \IS of bounded clique-width.
We note that \IS is chosen because \SAT, \ThreeSAT, and \MaxTwoSAT are still $\NP$-complete when its primal graph is a clique
($\cw=2$).
Hence, these problems parameterized by tree-width and clique-width are not equivalent unless $\cclass{P}=\NP$.
We believe that we can obtain similar results for many other problems that can be solved efficiently on graphs of
bounded clique-width.

The third contribution of this paper is introducing a new parameterized complexity class $\NLk$
(Exactly Parameterized $\NL$) and showing the following complete problems:
\begin{theorem}\label{thm:pw:complete}
\SAT, \ThreeSAT, \MaxTwoSAT, and \IS parameterized by path-width are $\NLk$-complete.
\end{theorem}
Intuitively, $\NLk$ is a class of parameterized problems that can be solved by a non-deterministic Turing machine with
the space of $k+O(\log n)$ bits.
For the precise definitions of $\NLk$ and $\NLk$-completeness, see Section~\ref{sec:pw}.
Flum and Grohe~\cite{DBLP:journals/iandc/FlumG03} introduced a similar class, called $\cclass{para-NL}$, that can be solved in $f(k)+O(\log n)$
space.
Although they showed that a trivial parameterization of an $\NL$-complete problem is $\cclass{para-NL}$-complete under
the standard parameterized reduction, this does not hold in our case because we use a different reduction to
define the complete problems.
If one of the $\NP$-complete problems can be solved in polynomial time, any problem in $\NP$ can be solved in
polynomial time.
Similarly, if one of the $\NLk$-complete problems can be solved in $O^*(c^k)$ time, any problem in $\NLk$ can be
solved in $O^*(c^k)$ time.
Since the class $\NLk$ contains many famous problems, such as \SetCover parameterized by the number of elements and
\TSP parameterized by the number of vertices, for which no $O^*((2-\epsilon)^n)$-time algorithms are known, our result
implies that we can use the hardness of not only \SAT but also these problems to establish the hardness of the problems parameterized by path-width.

\subsection{Overview of Decomposition-based Reductions}
We explain the basic idea of decomposition-based reductions.
Although we deal with three different decompositions in this paper, the basic idea is the same.
We believe that the same idea can be used to other decompositions such as branch-decomposition.

A decomposition can be seen as a collection of sets forming a tree.
For example, tree-decomposition is a collection of bags forming a tree and clique-decomposition is a collection of
labels forming a tree.
First, for each node $i$ of a decomposition tree, we create gadgets as follows:
(1) for each element $x$ in the corresponding set $X_i$, create a \emph{path-like} gadget $x_i$ that expresses the
\emph{state} of the element (e.g. the value of the variable $x$ for the case of \SAT),
and (2) create several gadgets to solve \emph{subproblem} corresponding to this node (e.g. simulate clauses inside $X_i$
for the case of \SAT).
Then, for each node $c$, its parent $p$, and each common element $x\in X_c\cap X_p$, by connecting the tail of
$x_c$ and the head of $x_p$, we establish \emph{local consistency}.
From the definition of the decomposition, this leads to \emph{global consistency}.
Since the obtained graph has a \emph{locality}, it has a small width.
We may need additional tricks to establish local consistency without increasing the width.

\subsection{Organization}
\begin{figure}[!t]
  \centering
  \includegraphics[scale=0.35]{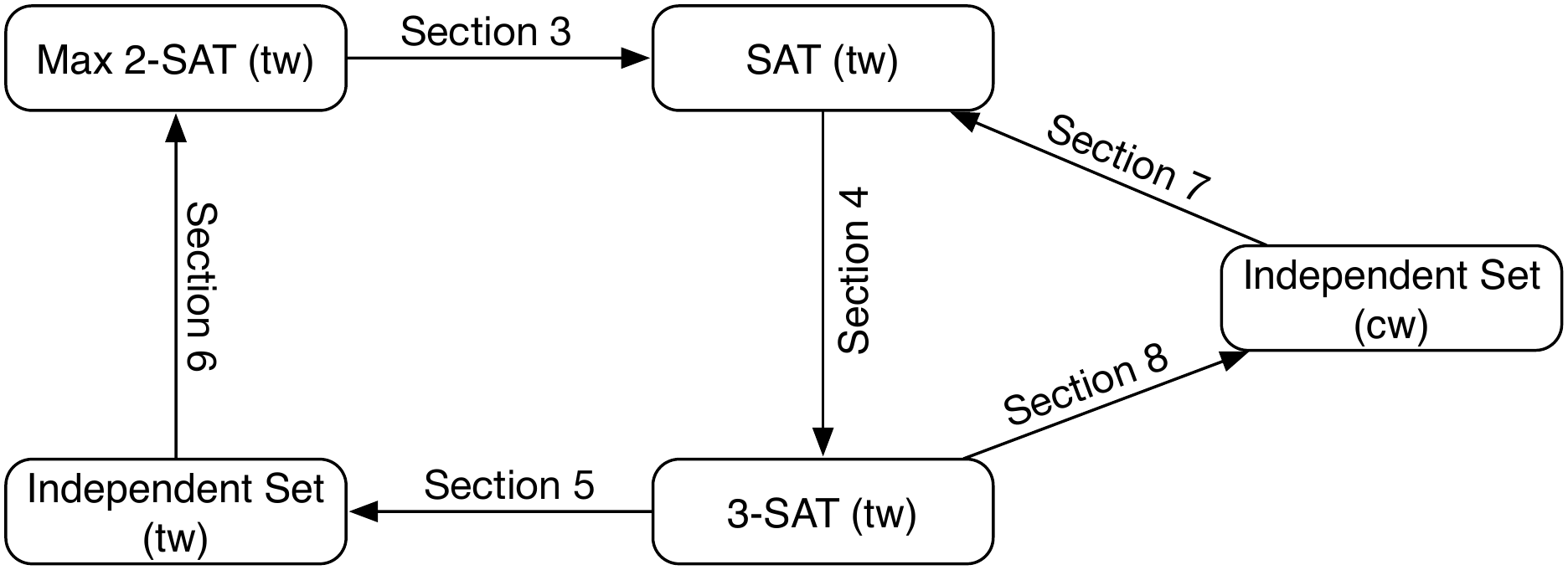}
  \caption{Reductions given in this paper}
  \label{fig:intro:tw}
\end{figure}
The rest of the paper is organized as follows.
In Section~\ref{sec:pre}, we introduce definitions and basic lemmas often used in this paper.
In Section~\ref{sec:tw:max2sat_sat}, we give a tree-width preserving reduction from \MaxTwoSAT to \SAT.
The reduction is rather simple but contains an essential idea of tree-decomposition-based reductions.
The other reductions are given in
Sections~\ref{sec:tw:sat_3sat}~-~\ref{sec:tw-3sat-to-cw-is} (see Figure~\ref{fig:intro:tw}).
In Section~\ref{sec:pw}, we introduce $\NLk$ and show that SAT parameterized by path-width is $\NLk$-complete.
%Due to space limitations, several reductions and proofs are omitted.
%These will be presented in the full version of this paper.

%!TEX root = ../paper.tex

\section{Preliminaries}\label{sec:pre}

%In this section, we give definitions and lemmas used in this paper.

For an integer $k$, we denote the set $\{1,2,\ldots,k\}$ by $[k]$ and the set $\{0,1,\ldots,k-1\}$ by $[k]'$.
Let $G=(V,E)$ be an undirected graph.
We denote the \emph{degree} of a vertex $v$ as $d_G(v)$.
We denote the \emph{neighborhood} of a vertex $u$ by $N_G(u)=\{v\in V\mid
\{u,v\}\in E\}$, and the \emph{closed neighborhood} of $u$ by
$N_G[u]=N_G(u)\cup\{u\}$.
Similarly, we denote the neighborhood of a subset $S\subseteq V$ by
$N_G(S)=\bigcup_{v\in S}N_G(v)\setminus S$, and the closed neighborhood by $N_G[S]=N_G(S)\cup S$.
We drop the subscript $G$ when it is clear from the context.
For a subset $S\subseteq V$,
let $G[S]=(S,\{\{u,v\}\in E\mid u\in S, v\in S\})$ denote the \emph{subgraph induced by $S$}.
For a vertex $v\in V$, let $G/v$ denote the graph obtained by removing $v$ and making the neighbors of $v$ form a clique.
We call this operation \emph{eliminating} $v$.
Similarly, for a subset $S\subseteq V$, we denote by $G/S$ the graph obtained by removing $S$ and making the neighbors of $S$ form a clique.

% A \emph{perfect matching} of a graph $G=(V,E)$ is a set $M\subseteq E$ of size $\frac{|V|}{2}$ such that no two edges in $M$ share a common vertex.
% \PM is the problem in which, given a graph, the objective is counting the number of perfect matchings of $G$.
% A \emph{dominating set} of a graph $G=(V,E)$ is a set $D\subseteq V$ such that any vertex in $V\setminus D$ is adjacent
% to some vertex in $D$.
% Given a graph $G$ and an integer $k$, \DOM is a problem to determine whether there exists a dominating set of size at most $k$.

A \emph{tree-decomposition} of a graph $G=(V,E)$ is a pair $(T,\chi)$, where $T=(I,F)$ is a tree and
$\chi=\{X_i\subseteq V\mid i\in I\}$ is a collection of subsets of vertices (called \emph{bags}), with the following
properties:
\begin{enumerate}
  \itemsep=0pt
  \item $\bigcup_i X_i=V$.
  \item For each edge $uv\in E$, there exists a bag that contains both of $u$ and $v$.
  \item For each vertex $v\in V$, the bags containing $v$ form a connected subtree in $T$.
\end{enumerate}
In order to avoid confusion between a graph and its decomposition tree $T$, we call a vertex of the tree a
\emph{node}, and an edge of the tree an \emph{arc}.
We identify a node $i\in I$ of the tree and the corresponding bag $X_i$.
The \emph{width} of a tree-decomposition is the maximum of $|X_i|-1$ over all nodes $i\in I$.
The \emph{tree-width} of a graph $G$, $\tw(G)$, is the minimum width among all the possible tree-decompositions of
$G$.
%If a graph contains a clique, from the definition of tree-decomposition, there must be a bag containing all the
% vertices of the clique.
%Thus, the tree-width is at least the size of the clique minus one.

A \emph{nice tree-decomposition} is a tree decomposition such that the root bag $X_r$ is an empty set and each node $i$ is one of the following types:
\begin{enumerate}
  \itemsep=0pt
  \item Leaf: a leaf node with $X_i=\emptyset$.
  \item Introduce($v$): a node with one child $c$ such that $X_i=X_c\cup \{v\}$ and $v\not\in X_c$.
  \item Introduce($uv$): a node with one child $c$ such that $u,v\in X_i=X_c$. We require that this node
  appears exactly once for each edge $uv$ of $G$.
  \item Forget($v$): a node with one child $c$ such that $X_i=X_c\setminus \{v\}$ and $v\in X_c$. From the
  definition of tree-decompositions, this node appears exactly once for each vertex of $G$.
  \item Join: a node with two children $l$ and $r$ with $X_i=X_l=X_r$.
\end{enumerate}
Any tree-decomposition can be easily converted into a nice tree-decomposition of the same width in polynomial time by
inserting intermediate bags between each adjacent bags.
Thus, in this paper, we use nice tree-decompositions to make discussions simple.

A (nice) \emph{path-decomposition} is a (nice) tree-decomposition $(T,\chi)$ such that the decomposition tree $T=(I,F)$
is a path.
%For convenience, we assume that $I=[N]$ for some integer $N$ and $F=\{(i,i+1)\mid i\in [N-1]\}$.
The \emph{path-width} of a graph $G$, $\pw(G)$, is the minimum width among all the possible path-decompositions of
$G$.

In order to prove the upper bound on tree-width, we will often use the following lemmas.
%Due to space limitations, proofs of these lemmas are omitted.
%are deferred to Appendix~\ref{sec:proofs}.
\begin{lemma}[Arnborg~\cite{DBLP:journals/bit/Arnborg85}]\label{lem:tw:eliminate}
For a graph $G=(V,E)$ and a vertex $v\in V$, $\tw(G)\leq\max(d(v),\tw(G/v))$.
Moreover, if we are given a tree-decomposition of $G/v$ of width $w$, we can construct a tree-decomposition of $G$ of
width $\max(d(v),w)$ in linear time.
\end{lemma}
\begin{proof}
Let $T=(I,F)$ be a tree-decomposition of $G/v$ of width $w$.
Since the neighbors $N(v)$ form a clique in $G/v$, there exists a node $i\in I$ such that the bag $X_i$
contains $N(v)$.
Therefore, by creating a node $j$ with $X_j=N[v]$ and adding an arc $ij$, we can obtain a tree-decomposition of $G$.
The width of this tree-decomposition is $\max(|X_j|-1,w)=\max(d(v),w)$.
\end{proof}
\begin{lemma}\label{lem:tw:eliminate_subset}
For a graph $G=(V,E)$ and a vertex subset $S\subseteq V$, $\tw(G)\leq\max(|N[S]|-1,\tw(G/S))$.
\end{lemma}
\begin{proof}
Let $S=\{v_1,...,v_k\}$.
We eliminate each vertex of $S$ one by one.
We denote the graph after the $i$-th elimination by $G_i=((G/v_1)/v_2)\ldots/v_i$.
By eliminating $v_i$ from $G_{i-1}$, we obtain $\tw(G_{i-1})\leq\max(d_{G_{i-1}}(v_i),\tw(G_i))$.
Since $N_{G_{i-1}}(v_i)\subseteq N_G[S]\setminus\{v_i\}$, we have $\tw(G)=\tw(G_0)\leq\max(|N[S]|-1, \tw(G_k))$.
Because $G_k$ is a subgraph of $G/S$ and any tree-decomposition of a graph is also a tree-decomposition of its
subgraph, we obtain $\tw(G)\leq\max(|N[S]|-1,\tw(G/S))$.
\end{proof}
\begin{lemma}\label{lem:tw:eliminate_matching}
Let $X$ and $Y$ be disjoint vertex sets of a graph $G$ such that for each vertex $x\in X$, $|N(x)\cap Y|\leq 1$.
Then, $\tw(G)\leq\max(|N[X]\setminus Y|,\tw(G/X))$.
\end{lemma}
\begin{proof}
Let $X=\{x_i\mid i\in[k]\}$ and $U=N(X)\setminus Y$.
For an integer $i$, we denote the vertex set $\{x_j\mid j\in [i]\}$ by $X_i$.
We eliminate each vertex of $X$ one by one.
After eliminating vertices $X_{i-1}$, $x_i$ can be adjacent only to vertices in
$(X\setminus X_i)\cup\{Y\cap N(X_i)\}\cup U$.
Since $|Y\cap N(X_i)|\leq i$, we have $d(x_i)\leq |X|-i+i+|U|=|X|+|U|=|N[X]\setminus Y|$.
By iteratively applying Lemma~\ref{lem:tw:eliminate}, we obtain $\tw(G)\leq\max(|N[X]\setminus Y|,\tw(G/X))$.
\end{proof}
\begin{lemma}\label{lem:tw:eliminate_layer}
Let $\{S_i\mid i\in [d]\}$ be a family of disjoint vertex sets of a graph $G$ such that each set has size at most $k$
and there are no edges between $S_i$ and $S_j$ for any $|i-j|>1$.
Then, $\tw(G)\leq\max(2k+|N(S)|-1,\tw(G/S))$, where $S=\bigcup_{i\in [d]}S_i$.
\end{lemma}
\begin{proof}
Let $U=N(S)$.
We eliminate each vertex set $S_i$ one by one.
After eliminating vertex sets $\{S_j\mid j\in [i-1]\}$, it holds that $N(S_i)\subseteq S_{i+1}\cup U$.
Thus, we have $|N[S_i]|\leq 2k+|U|$.
By iteratively applying Lemma~\ref{lem:tw:eliminate_subset}, we obtain $\tw(G)\leq\max(2k+|N(S)|-1,\tw(G/S))$.
\end{proof}
For a vertex set $S$, if we can obtain $\tw(G)\leq\max(d,\tw(G/S))$ by applying one of these lemmas, we say that the elimination has \emph{degree} $d$.
If we can reduce a graph $G$ into a graph $G'$ by a series of eliminations of degree at most $d$, we can obtain
$\tw(G)\leq\max(d,\tw(G'))$.

Let $x$ be a Boolean variable.
We denote the negation of $x$ by $\overline{x}$.
A \emph{literal} is either a variable or its negation,
and a \emph{clause} is a disjunction of several literals $l_1,\dots,l_k$,
where $k$ is called the \emph{length} of the clause.
We call a clause of length $k$ a \emph{$k$-clause}.
A \emph{CNF} is a conjunction of clauses.
If all the clauses have length at most $k$, it is called a \emph{k-CNF}.
We say that a CNF on a variable set $X$ is \emph{satisfiable} if there is an assignment to $X$ that makes the CNF true.
($k$-)\SAT is a problem in which, given a variable set $X$ and a ($k$-)CNF $\C$, the objective is to determine whether $\C$ is satisfiable or not.
% \problem{\#($k$-)SAT} is a problem of counting the number of satisfying assignments of $\C$.
\MaxTwoSAT is a problem in which, given a variable set $X$, a 2-CNF $\C$, and an integer $k$, the objective is to determine whether there exists an assignment that satisfies at least $k$ clauses in $\C$.

Let $\C$ be a CNF on variables $X$.
The \emph{primal graph} of $\C$ is the graph $G=(X,E)$ such that there exists an edge between two vertices if and only if their corresponding variables appear in the same clause.
For readability, we identify a variable or a literal as the corresponding vertex in the primal
graph.
That is, we may use the same symbol $x$ to indicate both a variable in a CNF and the corresponding vertex in the primal graph, and both literals $x$ and $\overline{x}$ correspond to the identical vertex in the primal graph.
%By a decomposition of a CNF, we mean a decomposition of its primal graph.
For a CNF $\C$, we slightly change the definition of the nice tree-decomposition as follows:
\begin{enumerate}
  \itemsep=0pt
  \item[$3'$.] Introduce($C$): an internal node with one child $c$ such that $X_i=X_c$ and all the variables in $C$ are
  in $X_i$. We require that this node appears exactly once for each clause $C\in\C$.
\end{enumerate}
Note that because the variables in the same clause form a clique in the primal graph, there always exists a bag that
contains all of them.
%Thus, we can always construct a nice tree-decomposition of the same width in polynomial time.

In our reductions, we will use a binary representation of an integer.
Let $\{a_1,a_2,\ldots,a_M\}$ be Boolean variables.
We denote the integer $\sum_{i\in[M], a_i=\true}2^{i-1}$ by $(a_1 a_2 \ldots a_M)_2$, or $(a_*)_2$ for short.
For readability, we will frequently use (arithmetic) constraints such as $(a_*)_2=(b_*)_2+(c_*)_2$.
% Note that such an arithmetic constraint can be expressed using clauses.
Note that any arithmetic constraint on $M$ variables can be trivially simulated by at most $2^M$ $M$-clauses.
Thus, if $M$ is logarithmic in the input size, the number of required clauses is polynomial in the input size.

%!TEX root = ../paper.tex

\section{Tree-width preserving reduction from \MaxTwoSAT to \SAT}\label{sec:tw:max2sat_sat}
Let $(X,\C=\{C_1,\ldots,C_m\},k)$ be an instance of \MaxTwoSAT.
We want to construct an instance $(X',\C')$ of \SAT such that $\C'$ is satisfiable if and only if at least $k$ clauses of $\C$ can be
satisfied.
Let $M=\left\lceil\log{(m+1)}\right\rceil$.
In the following reductions, we will use arithmetic constraints on $O(M)$ variables, which can be simulated by $\poly(m)$ clauses.

% First, we explain a naive reduction from \MaxTwoSAT to \SAT that does not preserve tree-width.
% For each clause $C_i=(x\vee y)\in\C$, we create a new variable $w_i$ and a constraint $w_i\Leftrightarrow(x\vee y)$.
% The variable $w_i$ represents whether the clause $C_i$ is satisfied.
% For each $j\in[M]$, we create a variable $s_{0,j}$ and a clause $(\overline{s_{0,j}})$.
% For each $i\in[m]$ and $j\in[M]$, we create a variable $s_{i,j}$.
% Then, we will insert clauses so that $(s_{i,*})_2$ represents the number of satisfied
% clauses in $\{C_1,C_2,\ldots,C_i\}$.
% This can be done by inserting a constraint of
% $(s_{i,*})_2=(s_{i-1,*})_2+(w_i)_2$.
% Finally, we create a constraint that $(s_{m,*})_2\geq k$.
% Now, we have obtained an instance $\C'$ of polynomial size.
% We can easily check that $\C'$ is satisfiable if and only if at least $k$ clauses of $\C$ can be satisfied.

% Let $\tw$ be the tree-width of $\C$.
% We want to show that $\C'$ has tree-width at most $\tw+O(\log m)$.
% We note that the additive $O(\log m)$ factor is allowed because $O^*(2^{\alpha(\tw+O(\log
% m))})=O^*(2^{\alpha\tw}\poly(m))=O^*(2^{\alpha\tw})$.
% Unfortunately, however, we cannot obtain such an upper bound for the naive reduction above.
% This is because the variables $w_i$'s break the structure of the tree-decomposition and blow up the tree-width of the resulting graph.
% To resolve this issue, we determine the order of adding $w_i$'s based on the tree-decomposition.
% 
% Now, we introduce a tree-decomposition-based reduction.
Let $T=(I,F)$ be a given nice tree-decomposition of width $\tw$.
We will create an instance of \SAT whose tree-width is at most $\tw+O(\log m)$.
We note that the additive $O(\log m)$ factor is allowed because $O^*(2^{\alpha(\tw+O(\log
m))})=O^*(2^{\alpha\tw}\poly(m))=O^*(2^{\alpha\tw})$.
For each node $i\in I$, we create variables $\{x_i\mid x\in X_i\}\cup\{s_{i,j}\mid j\in [M]\}\cup\{w_i\}$.
The value $(s_{i,*})_2$ will represent the number of satisfied clauses in the subtree rooted at $i$.
For each node $i$ and its parent $p$,
we create a constraint $x_i=x_p$ for each variable $x\in X_i\cap X_p$.
Because the nodes containing the same variable form a connected subtree in $T$, these constraints ensure that for any variable
$x \in X$, all the variables $\{x_i\mid x\in X_i\}$ take the same value.
For each node $i$, according to its type, we do as follows:
\begin{enumerate}
  \item Leaf: create a clause $(\overline{s_{i,j}})$ for each $j\in[M]$.
  \item Introduce($v$): create a constraint $s_{i,j}=s_{c,j}$ for each $j\in [M]$.
  \item Introduce($x\vee y$): create a constraint $w_i\Leftrightarrow(x_i\vee y_i)$ and a
  constraint $(s_{i,*})_2=(s_{c,*})_2+(w_i)_2$.
  \item Forget($v$): create a constraint $s_{i,j}=s_{c,j}$ for each $j\in [M]$.
  \item Join: create a constraint $(s_{i,*})_2=(s_{l,*})_2+(s_{r,*})_2$.
\end{enumerate}
Finally for the root node $r$, we create a constraint $(s_{r,*})_2\geq k$.
Now, we have obtained an instance $(X',\C')$ of polynomial size.
We note that, from the definition of a nice tree-decomposition, there exists exactly one Introduce($C$) node for each
clause $C\in\C$.
Thus, the sum $\sum_{i\in I} (w_i)_2$, which is equal to $(s_{r,*})_2$, represents the number of satisfied
clauses.
Therefore, $\C'$ is satisfiable if and only if at least $k$ clauses of $\C$ can be satisfied.
Finally, we show that the reduction preserves the tree-width.

\begin{lemma}
$\C'$ has tree-width at most $\tw+O(\log m)$.
\end{lemma}
\begin{proof}
We will prove the bound by reducing the primal graph of $\C'$ into an empty graph by a
series of eliminations of degree at most $\tw+O(\log m)$.
For a node $i$, let $Y_i$ denote the vertex set $\{x_i\mid x\in X_i\}$ and $V_i$ denote the vertex set
$Y_i\cup\{w_i\}\cup\{s_{i,j}\mid j\in[M]\}$.
Starting from the primal graph of $\C'$ and the given tree-decomposition $T$ of $\C$, we eliminate the vertices as
follows.
First, we choose an arbitrary leaf $i$ of $T$.
Then, we eliminate all the vertices of $V_i$ in a certain order,
which will be described later.
Finally, we remove $i$ from $T$ and repeat the process until $T$ becomes empty.

Let $i$ be a leaf and $p$ be its parent.
If $i$ is the only child of $p$, we have $N(V_i)\subseteq V_p$.
Thus, the eliminations of $V_i$ can create edges only inside $V_p$.
If $p$ has another child $q$, we have $N(V_i)\subseteq V_p\cup\{s_{q,j}\mid j\in [M]\}$.
Thus, the eliminations of $V_i$ can create edges only inside $V_p\cup\{s_{q,j}\mid j\in [M]\}$.
Therefore, after processing each node, we can ensure that the edges created by previous eliminations
are only inside $V_i\cup\{s_{c,j}\mid c\text{ is a child of }i\text{ and }j\in [M]\}$ for each node $i$.

Now, we describe the details of the eliminations.
Let $i$ be the current node to process.
If $i$ is the root, the number of remaining vertices is $O(\log m)$.
Thus, the elimination of these vertices has degree $O(\log m)$.
Otherwise, let $p$ be the parent of $i$.
First, we eliminate the vertices $Y_i$.
Because each vertex of $Y_i$ is adjacent to at most one vertex of $Y_p$,
Lemma~\ref{lem:tw:eliminate_matching} gives the elimination of degree $|N[Y_i]\setminus Y_p|\leq |V_i|\leq \tw+O(\log m)$.
Then, we eliminate the remaining vertices $V_i\setminus Y_i$.
If $i$ is the only child of $p$, let $V_q=Y_q=\emptyset$, and otherwise, let $q$ be the another child of $p$.
By applying Lemma~\ref{lem:tw:eliminate_subset}, we obtain the elimination of degree $|N[V_i\setminus
Y_i]|-1\leq|V_i\setminus Y_i|+|V_p|+|V_q\setminus Y_q|\leq\tw+O(\log m)$.
\end{proof}

%!TEX root = ../paper.tex

\section{Tree-width preserving reduction from \SAT to \ThreeSAT}\label{sec:tw:sat_3sat}
Let $(X,\C=\{C_1, \ldots, C_m\})$ be an instance of \SAT and $\tw$ be its tree-width.
We can use the standard reduction from \SAT to \ThreeSAT:
replacing each clause $(x_1\vee\ldots\vee x_k)$ with clauses $(x_1\vee x_2\vee y_1), (\overline{y_1}\vee
x_3\vee y_2),\ldots,(\overline{y_{k-3}}\vee x_{k-1}\vee x_k)$.

Now, we show that the tree-width of the obtained 3-CNF is at most $\tw+2$.
For each clause $(x_1\vee\ldots\vee x_k)$ of the original CNF, we have created $k-3$ new variables
$Y=\{y_1,y_2,\ldots,y_{k-3}\}$.
Let $S_i=\{y_i\}$ for $i\in [k-3]$.
Since there is no edge between $S_i$ and $S_j$ for $|i-j|>1$, from Lemma~\ref{lem:tw:eliminate_layer}, we obtain the
elimination of $Y$ of degree $2\times 1+|N(S)|-1=k+1$.
Since variables in the same clause form a clique in the primal graph, we have $\tw\geq k-1$.
Thus, the elimination has degree at most $\tw+2$.
After applying the above elimination to all the clauses, the graph coincides with the primal graph of $\C$.
Therefore, the tree-width of the obtained 3-CNF is at most $\tw+2$.

%!TEX root = ../paper.tex

\section{Tree-width preserving reduction from \ThreeSAT to \IS}\label{sec:tw:3sat_is}
An \emph{independent set} of a graph $G=(V,E)$ is a set $S\subseteq V$ such that $G[S]$ has no edges.
\IS is the problem in which, given a graph $G$ and an integer $k$, the objective is to determine whether there exists an independent set of $G$ with size at least $k$.

Let $(X,\C=\{C_1,\ldots,C_m\})$ be an instance of \ThreeSAT.
% We denote the set of variables in $\C$ by $X$.
We want to construct an instance $(G,k)$ of \IS with essentially the same tree-width such that $G$ has an independent
set of size at least $k$ if and only if $\C$ is satisfiable.
Actually, in our reductions, we choose $k$ so that any independent set has size at most $k$.
In the following reductions, we will use two gadgets depicted in Figure~\ref{fig:tw:3sat_is:gadgets}.

\begin{figure}[h]
  \centering
  \includegraphics[scale=0.33]{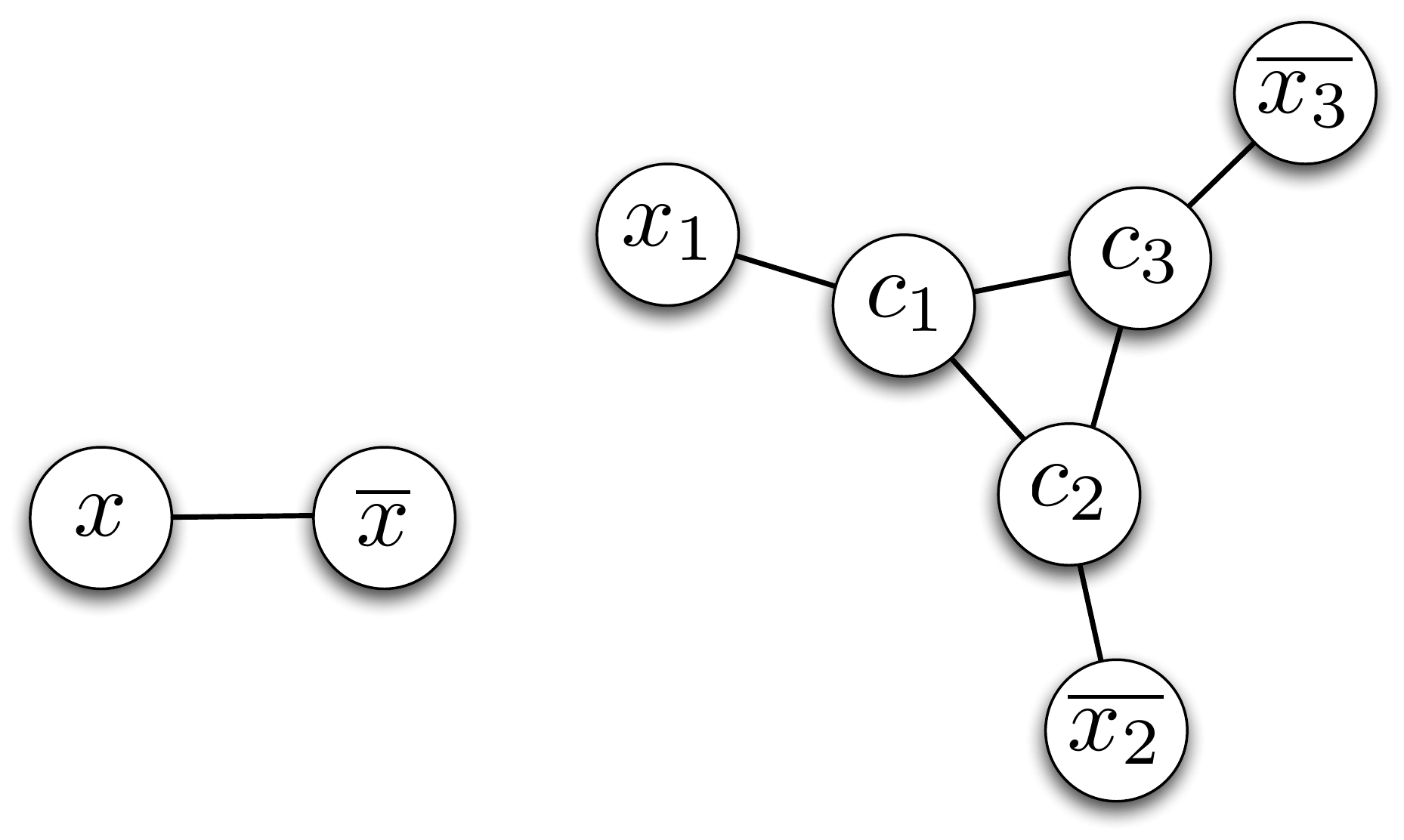}
  \caption{The variable gadget for a variable $x$ and the clause gadget for a clause $(\overline{x_1} \vee x_2 \vee x_3)$}
  \label{fig:tw:3sat_is:gadgets}
\end{figure}

A \emph{variable gadget} of a variable $x$ consists of two vertices $x$ and $\overline{x}$ connected by an edge.
Any independent set can contain at most one of $x$ and $\overline{x}$.
By choosing $k$ properly, we ensure that any independent set of size $k$ contains exactly one of them.
This gadget will represent whether a variable $x$ is assigned true (the vertex $x$ is in the independent set) or false
(the vertex $\overline{x}$ is in the independent set).

A \emph{clause gadget} of a clause $C=(x_1\vee x_2\vee\ldots\vee x_d)$ consists of $d$ vertices $\{c_i\mid i\in [d]\}$ forming a clique ($x_1,\ldots,x_d$ are literals rather than variables).
By choosing $k$ properly, we ensure that any independent set of size $k$ contains exactly one of them.
We call the operation of creating a clique $\{c_i\mid i\in[d]\}$ and
inserting edges $\{c_i\overline{x_i}\mid i\in [d]\}$ \emph{creating a clause gadget $C$}.
If an independent set contains one vertex from the clause gadget, at least one of the vertices $\{\overline{x_i}\mid
i\in [d]\}$ are not in the independent set.
By our choice of $k$, we ensure that at least one of $\{x_i\mid i\in [d]\}$ must be in the independent set.
Therefore, it acts as a clause $C$.
% We note that we will create clause gadgets for clauses that simulate arithmetic constraints among variables.
% For such clause gadgets, we may have $d > 3$.

% Instead of clauses, we may use constraints on $O(\log \tw)$ variables, which can be simulated by $O(\poly(\tw))$ clauses of length $O(\log \tw)$.
% Note that, each clause gadget created to simulate a constraint on variables $Y$ is adjacent only to the vertices in the
% variable gadgets for $Y$.

First, we explain a naive reduction from \ThreeSAT to \IS that does not preserve tree-width.
For each variable $x\in X$, we create a corresponding variable gadget,
and  for each clause $(x\vee y\vee z)\in\C$, we create a corresponding clause gadget.
Finally, we set $k$ as the number of variable gadgets plus the number of clause gadgets.
Now, we have obtained an instance $(G,k)$ of \IS.
From our choice of $k$, if $G$ contains an independent set of size $k$, it must contain exactly one vertex
from each variable gadget and clause gadget.
Therefore $\C$ is satisfiable.
Conversely, if $\C$ is satisfiable, we can construct an independent set of size $k$ by choosing an appropriate vertex
from each gadget.

Let $\tw$ be the tree-width of $\C$.
We omit the proof but the above naive reduction increases the tree-width of $G$ to $2\tw+O(1)$.
This is because, instead of a single variable $x$, we need to keep two vertices $x$ and $\overline{x}$ of the
variable gadget in a bag.
Intuitively, in order to preserve tree-width, we can put only one of $x$ and $\overline{x}$ in a bag.
Our solution is forgetting and remembering the state of $x$ and $\overline{x}$ along the tree-decomposition.

Now, we explain our tree-decomposition-based reduction.
Let $M=\left\lceil\log{(\tw+2)}\right\rceil$.
We will construct a graph with tree-width at most $\tw+O(\log \tw)$.
As we discussed before, the additive $O(\log \tw)$ factor is allowed.
Let $T=(I,F)$ be a given nice tree-decomposition of width $\tw$.
For each node $i\in I$, we create a variable gadget for each of $\{x_i\mid x\in X_i\}$.
If $i$ is an Introduce($x\vee y\vee z$) node, we create a clause gadget for $(x_i\vee y_i\vee z_i)$.
If $i$ is not the root, let $p$ be its parent and $P_i$ be the set $X_i\cap X_p$.
Then, for each $x\in P_i$, we connect $\overline{x_i}$ and $x_p$ by an edge.
We want to ensure that for any independent set $S$ of size $k$, $x_i$ is in $S$ if and only if $x_p$ is in $S$.
If $\overline{x_i}$ is in $S$, $x_p$ cannot be in $S$, and therefore $\overline{x_p}$ must be in $S$.
On the other hand, even if $x_i$ is in $S$, $\overline{x_p}$ can be in $S$.
In order to avoid such a situation, we will create a gadget to count the number of vertices in $(\{x_i\mid x\in
P_i\}\cup\{\overline{x_p}\mid x\in P_i\})\cap S$ (this is the most interesting part of our reduction).
Because $x_i\not\in S$ implies $\overline{x_p}\in S$, the number is always at least $|P_i|$,
and if (and only if) the number is exactly $|P_i|$, it holds that $x_i\in S\Leftrightarrow x_p\in S$ for any $x\in P_i$.
Since the nodes containing the same variable form a connected subtree, this ensures that for any independent set of
size $k$ and for any variable $x$, all the vertices $\{x_i\mid x\in X_i\}$ are in $S$ or none of them are in $S$.
By using the binary encoding, the number can be expressed by $O(\log \tw)$ variables.
Thus, we can make the gadget to increase the tree-width only by $O(\log \tw)$.

We will construct such a gadget by using the following gadget.
Let $U=\{u_1,\ldots,u_d\}$ be a set of vertices.
A \emph{counting gadget of $U$} consists of the following $d+1$ layers of variable gadgets connected by clause gadgets.
For each $a\in [d]$, the $a$-th layer consists of a variable gadget for $y_a$ and variable gadgets for each of
$\{s_{a,j}\mid j\in [M]\}$.
The last layer consists of variable gadgets for each of $\{s_{d+1,j}\mid j\in [M]\}$.
Then, for each $j\in [M]$, we create a clause gadget for $(\overline{s_{1,j}})$, and for each $a\in [d]$, we create
clause gadgets simulating an arithmetic constraint $(s_{a+1,*})_2=(s_{a,*})_2+(y_a)_2$.
Finally, for each $a\in [d]$, we connect $u_a$ and $\overline{y_a}$ by an edge.
For an independent set $S$, the number $(s_{d+1,*})_2$ in the last layer represents the size of
$\{y_a\mid a\in [d]\}\cap S$.
Since $u_a\in S$ implies $y_a\in S$, the number is at least the size of $U\cap S$.

\begin{figure}
  \centering
  \includegraphics[width=0.8\hsize]{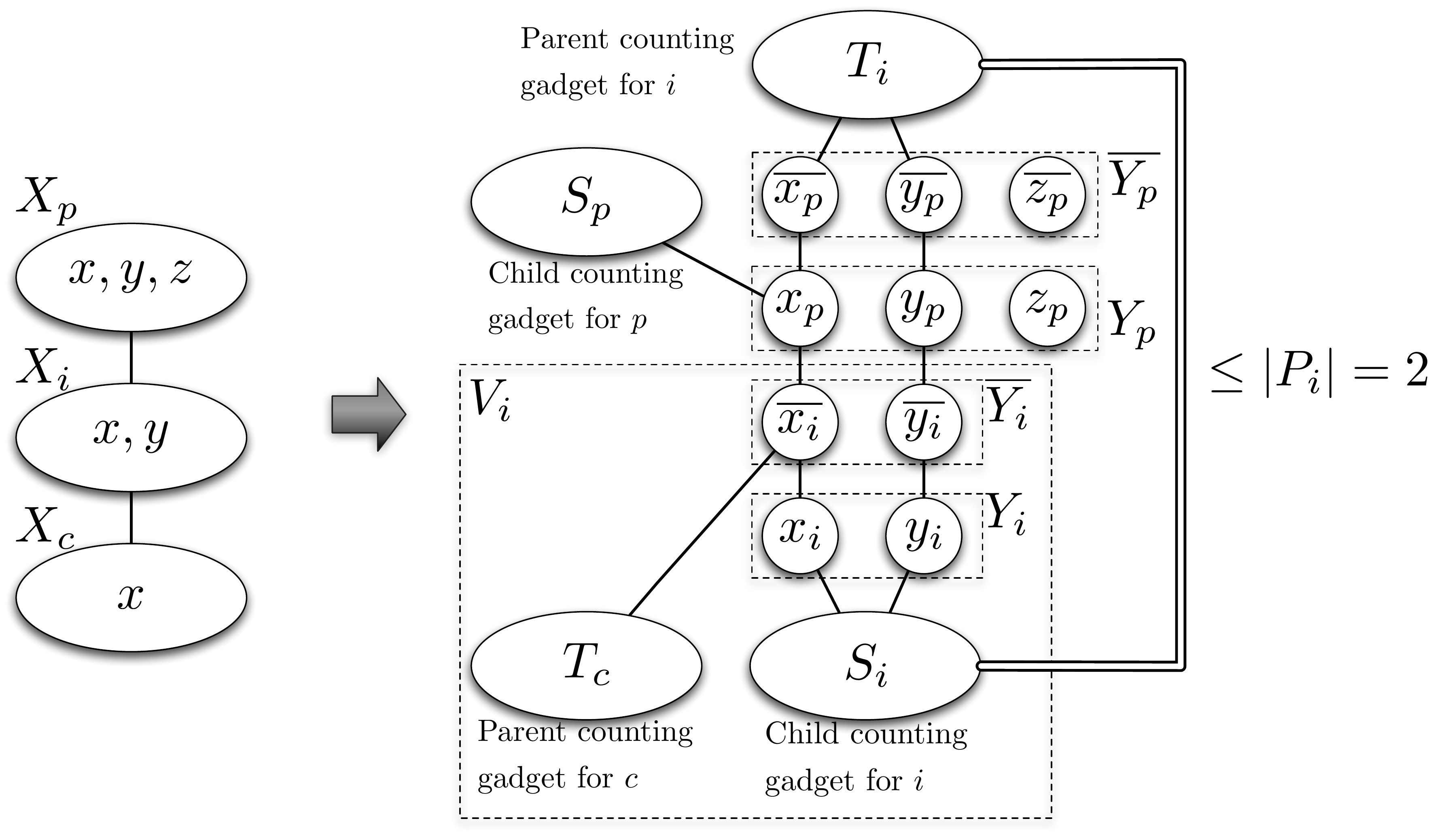}
  \caption{Reduction from \ThreeSAT to \IS}
  \label{fig:tw:3sat_is:counting}
\end{figure}

Now, we construct the gadget (see Figure~\ref{fig:tw:3sat_is:counting}).
First, we construct a counting gadget for the set $\{x_i\mid x\in P_i\}$, called a \emph{child counting gadget for $i$}.
Then, we construct a counting gadget for the set $\{\overline{x_p}\mid x\in P_i\}$, called a \emph{parent counting gadget for $i$}.
Finally, we create clause gadgets simulating the arithmetic constraint that the sum of the numbers represented by the last layers of these two counting gadgets must be at most $|P_i|$.
As we discussed before, the size $|(\{x_i\mid x\in P_i\}\cup\{\overline{x_p}\mid x\in P_i\})\cap S|$ is always at
least $|P_i|$ and becomes exactly $|P_i|$ if and only if $x_i\in S\Leftrightarrow x_p\in S$ holds for any $x\in P_i$.
Since the sum is at least the size $|(\{x_i\mid x\in P_i\}\cup\{\overline{x_p}\mid x\in P_i\})\cap S|$, the
constraint that the sum is at most $|P_i|$ implies that $x_i\in S\Leftrightarrow x_p\in S$ for any $x\in P_i$.

Now, we have obtained a graph $G$ of polynomial size and we set $k$ as the number of variable gadgets plus the number of
clause gadgets.
From our construction, for any independent set $S$ of size $k$ and a variable $x \in X$, all the vertices $\{x_i\mid i \in I\text{ s.t. }x \in X_i\}$ are in $S$ or none
of them are in $S$. %, and all the clauses must be satisfied.
Thus, if $G$ has an independent set of size $k$, $\C$ is satisfiable.
Conversely, if $\C$ is satisfiable, by taking an appropriate vertex from each gadget, we can obtain an independent set
of size $k$.
Finally, we show that the reduction preserves the tree-width.

\begin{lemma}
$G$ has tree-width at most $\tw+O(\log \tw)$.
\end{lemma}
\begin{proof}
We will prove the bound by reducing $G$ into an empty graph by a series
of eliminations of degree at most $\tw+O(\log \tw)$.
Starting from $G$ and the given tree-decomposition $T$ of $\C$, we eliminate the vertices as
follows.

First, for each clause gadget other than the clause gadgets for $C \in \C$ (created when processing the Introduce($C$) node), we eliminate its vertices $S$.
Since the size of $N[S]$ is $O(\log\tw)$ and no two vertices in different clause gadgets are adjacent, from
Lemma~\ref{lem:tw:eliminate_subset}, we obtain eliminations of degree $O(\log \tw)$.

For a node $i \in I$, let $Y_i$ and $\overline{Y_i}$ denote the vertex sets $\{x_i\mid x\in X_i\}$ and $\{\overline{x_i}\mid x\in X_i\}$, respectively. % , that belong to the variable gadgets.
If $i$ is an Introduce node, then let $C_i$ denote the set of vertices in the corresponding clause gadget, and otherwise, let $C_i$ be an empty set.
If $i$ is not the root and has a parent $p$, let $S_{i,a}$ be the
set of vertices in the variable gadgets of the $a$-th layer of the child counting gadget for $i$.
If $i$ is the root, we set $S_{i,a}$ as an empty set.
We denote the set of all the vertices of the child counting gadget by $S_i=\bigcup_{a\in [d+1]}S_{i,a}$, where
$d=|X_i\cap X_p|$.
Similarly, let $T_{i,a}$ be the
set of vertices in the variable gadgets of the $a$-th layer of the parent counting gadget for $i$ and $T_i=\bigcup_{a\in [d+1]}T_{i,a}$.
Let $V_i$ denote the union of $Y_i$, $\overline{Y_i}$, $C_i$, $S_i$, and $T_c$ for each child $c$ of $i$.
Now, we eliminate each $V_i$ as follows.

First, we choose an arbitrary leaf $i$ of the tree $T$.
Then, we eliminate all the vertices of $V_i$ in a certain order, which will be described later.
Finally, we remove $i$ from $T$ and repeat the process until $T$ becomes empty.

Since $N(V_i)\subseteq Y_p\cup T_{i,d+1}$ holds for a leaf $i$ and its parent $p$, where $d=|X_i\cap X_p|$,
the eliminations of $V_i$ can create edges only within $Y_p\cup T_{i,d+1}$.
Thus, after processing each node, we can ensure that the edges created by previous eliminations only connect vertices in the same vertex set $Y_p\cup T_{i,d+1}$ for some node $i$, its parent $p$, and $d=|X_i\cap X_p|$.

Now, we describe the details of the eliminations.
Let $i$ be the current node to process.
If $i$ is the root, the number of remaining vertices is $O(\log \tw)$.
Thus, the elimination of these vertices has degree $O(\log \tw)$.
Otherwise, let $p$ be the parent of $i$, $d=|X_i\cap X_p|$, and $J$ be the set of children of $i$ in the original
tree-decomposition.
We note that from the definition of nice tree-decompositions, the size of $J$ is at most two.
First, we eliminate $S_i$.
Since there are no edges between $S_{i,a}$ and $S_{i,b}$ for any $|a-b|>1$, from Lemma~\ref{lem:tw:eliminate_layer},
the elimination has degree $2(M+2)+|N(S_i)|-1=O(\log\tw)+|Y_i\cup T_{i,d+1}|\leq\tw+O(\log\tw)$.
Then, we eliminate $Y_i$.
Note that each vertex $x_i\in Y_i$ can be adjacent only to the vertex $\overline{x_i}\in\overline{Y_i}$, vertices in $Y_i\cup C_i\cup T_{i,d+1}$ (as we have eliminated $S_i$), and vertices in $T_{c,|X_c\cap X_i|+1}$ for a child $c\in J$ (as $\overline{x_c}$ is adjacent to $x_i$ and the path $T_{c,|X_c\cap X_i|+1}$-$S_c$-$x_c$-$\overline{x_c}$ is eliminated when processing $c$).
Hence by Lemma~\ref{lem:tw:eliminate_matching}, the elimination has degree $|N[Y_i]\setminus\overline{Y_i}|\leq|Y_i\cup C_i\cup T_{i,d+1}|+\sum_{c\in J}|T_{c,|X_c\cap X_i|+1}|\leq\tw+O(\log\tw)$.
Next, we eliminate $C_i$.
From Lemma~\ref{lem:tw:eliminate_subset}, the elimination has degree $N[C_i]-1\leq 5+|\overline{Y_i}\cup
T_{i,d+1}|+\sum_{c\in J}|T_{c,|X_j\cap X_i|+1}|\leq \tw+O(\log\tw)$.
Then, for each child $c\in J$, we eliminate $T_c$.
Since there are no edges between $T_{c,a}$ and $T_{c,b}$ for any $|a-b|>1$, from Lemma~\ref{lem:tw:eliminate_layer},
the elimination has degree $2(M+2)+|N(T_c)|-1=O(\log\tw)+|\overline{Y_i}\cup T_{i,d+1}|+\sum_{j\in
J}|T_{j,|X_j\cap X_i|+1}|\leq\tw+O(\log\tw)$.
Finally, we eliminate $\overline{Y_i}$.
Since each vertex $\overline{x_i}\in \overline{Y_i}$ can be adjacent only to the vertex $x_p\in Y_p$ and vertices in
$\overline{Y_i}\cup T_{i,d+1}$, from Lemma~\ref{lem:tw:eliminate_matching}, the elimination has degree
$|N[\overline{Y_i}]\setminus Y_p|\leq |\overline{Y_i}\cup T_{i,d+1}|\leq \tw+O(\log\tw)$.
\end{proof}

%!TEX root = ../paper.tex

\section{Tree-width preserving reduction from \IS to \MaxTwoSAT}\label{sec:tw:is_max2sat}
Let $(G=(V,E), k)$ be an instance of \IS.
We use the following naive reduction to make an instance $(X',\C',k')$ of \MaxTwoSAT.

For each vertex $v\in V$, we create a variable $x_v$ and add a clause $(x_v)$ of length one.
This variable represents whether a vertex $v$ is in an independent set or not.
Then, for each edge $uv\in E$, we create $|V|+1$ copies of a clause $(\overline{x_u}\vee\overline{x_v})$.
This clause simulates the constraint that at most one of $u$ and $v$ can be in an independent set.
Finally, we set $k'=|E|(|V|+1)+k$.

If there exists an independent set $S$ of size at least $k$, we can satisfy at least $k'$
clauses by setting $x_v=\true$ if and only if $v\in S$.
If there exists an assignment that satisfies $k'$ clauses, it must satisfy all the constraints
$(\overline{x_u}\vee\overline{x_v})$.
Thus, we can construct an independent set $S$ of size at least $k$ by taking $v\in S$ if and only if $x_v=\true$.
Because the primal graph of the obtained CNF $\C'$ is completely the same as the original graph $G$, they have the same
tree-width.

%!TEX root = ../paper.tex

\section{From Independent Set parameterized by clique-width to SAT parameterized by
tree-width}\label{sec:is-cw-to-sat-tw}

In this section, we show a reduction from \IS of bounded clique-width to \SAT of bounded tree-width.
We first define the notion of clique-width formally.
The \emph{clique-width} of a graph $G$ is the minimum number of labels needed to construct $G$ by means of the following four operations.
\begin{itemize}
\itemsep=0pt
\item Creation of a vertex $v$ with a label $i$ (denoted by $i(v)$).
\item Disjoint union of two labeled graphs $G$ and $H$ (denoted by $G \oplus H$).
\item Joining each vertex with label $i$ to each vertex with label $j$, where $i \neq j$ (denoted by $\eta_{i,j}$).
\item Renaming label $i$ to label $j$ (denoted by $\rho_{i \to j}$).
\end{itemize}
Every graph can be defined by an algebraic expression using these four operations. For instance,
a chordless path $P_4$ on four consecutive vertices $a,b,c,d$ can be defined as follows:
\[
  \eta_{3,2}(3(d) \oplus \rho_{3 \to 2}(\rho_{2 \to 1}(\eta_{3,2}(3(c) \oplus \eta_{2,1}(2(b) \oplus 1(a)))))).
\]
Such an expression is called a \emph{$k$-expression} if it uses at most $k$ different labels.
Thus, the clique-width of $G$, denoted by $\cw(G)$, is the minimum $k$ for which there exists a $k$-expression defining $G$.
For instance, from the above example, we conclude $\cw(P_4) \leq 3$.

It is known that $\cw(G) \leq 2^{\tw(G)}$ holds for any graph $G$~\cite{Corneil:2006ic}.
However, bounded clique-width does not imply bounded tree-width.
For example, the complete graph of $n$ vertices has tree-width $n-1$ and clique-width $2$.

Let $(G=(V,E), k)$ be an instance of \IS of $n$ vertices.
Let $\cw$ be the clique-width of $G$.
We want to construct an instance $(X,\mathcal{C})$ of SAT with tree-width $\cw + O(\log n)$ such that $\mathcal{C}$ is satisfiable if and only if there is an independent set of size $k$ in $G$.
Let $M = \lceil \log (n+1) \rceil$.

Let $O$ be the set of operations in the $\cw$-expression of $G$.
Note that the $\cw$-expression of $G$ can be represented as a tree, which we call the \emph{expression tree} of $G$, and we will often identify an operation and the corresponding node in the expression tree.
For each operation $o \in O$, we associate a subgraph $G_o$ constructed by performing operations in the subtree rooted at the operation $o$.
For each operation $o \in O$,
we introduce variables $\{o_i \mid i \in [\cw]\} \cup \{s_{o,i} \mid i \in [M]\}$.
For each $i \in [\cw]$, the variable $o_i$ represents whether vertices with label $i$ are chosen to be an independent set in $G_o$, and $\{s_{o,i} \mid i \in [M]\}$ represents the size of the independent set in $G_o$.

We add constraints as follows depending on the type of the operation $o$.
\begin{itemize}
\itemsep=0pt
\item $o = i(v)$: create a constraint $(s_{o,*})_2 = (o_i)_2$.
\item $o = c \oplus c'$: create two constraints $c_i\rightarrow o_i$ and $c'_i\rightarrow o_i$ for each $i \in
[\cw]$, and a constraint $(s_{o,*})_2 = (s_{c,*})_2 + (s_{c',*})_2$.
\item $o = \eta_{i, j}(c)$: create a constraint $o_i = c_i$ for each $i \in [\cw]$, a constraint $\overline{o_i} \vee
\overline{o_j}$, and a constraint $(s_{o,*})_2 = (s_{c,*})_2$.
\item $o = \rho_{i \to j}(c)$: create a constraint $o_k=c_k$ for each $k\in [\cw]\setminus\{i,j\}$ and three constraints
$(s_{o,*})_2 = (s_{c,*})_2$, $o_j = c_i \vee c_j$, and $(\overline{o_i})$.
\end{itemize}
Finally, for the root operation $o \in O$, we add a constraint $(s_{o,*})_2 \geq k$.
Note that for an operation $o=c\oplus c'$, the created constraints $c_i\rightarrow o_i$ and $c'_i\rightarrow o_i$
actually perform as a constraint $o_i=(c_i\vee c'_i)$.
This is because if there exists a satisfiable assignment for which both of $c_i$ and $c'_i$ are set to false but $o_i$
is set to true, we can obtain another satisfiable assignment by setting $o_i$ and the variables
connected by equality constraints to false.

Now we have obtained an instance $(X,\mathcal{C})$ of polynomial size.
As the above construction directly simulates the dynamic programming for solving \IS, $\mathcal{C}$ is satisfiable if and only if there is an independent set of size at least $k$.
Now we show that the tree-width of the instance $(X,\mathcal{C})$ has essentially the same clique-width of the graph $G$.
\begin{lemma}
  $\mathcal{C}$ has tree-width at most $\cw + O(\log n)$.
\end{lemma}
\begin{proof}
  We will prove the bound by reducing the primal graph of $\mathcal{C}$ into an empty graph by a series of eliminations
  of degree at most $\cw + O(\log n)$.
  For an operation $o$, let $Y_o$ denote the vertex set $\{o_i \mid i \in [\cw]\}$,
  and $V_i$ denote the vertex set $Y_o \cup \{s_{o,i} \mid i \in [M]\}$.

  Starting from the primal graph of $\mathcal{C}$ and the given $\cw$-expression of $G$, we eliminate the vertices as follows.
  First, we choose an arbitrary operation $o$ corresponding to a leaf in the expression tree.
  Then, we eliminate all the vertices of $V_i$ in a certain order, which will be described later.
  Finally, we remove $o$ from the expression tree and repeat the process until the expression tree becomes empty.

  Let $o$ be an operation corresponding to a leaf of the current expression tree, and $p$ be its parent.
  If $o$ is the only child of $p$, it holds that $N(V_o) \subseteq V_p$.
  Thus, the eliminations of $V_o$ can create edges only inside $V_p$.
  If $p$ has another child $q$, it holds that $N(V_o) \subseteq V_p \cup \{s_{q,i} \mid i \in [M]\}$.
  Thus, the eliminations of $V_i$ can create edges only inside $V_p \cup \{s_{q,i} \mid i \in [M]\}$.
  Therefore, after processing each node, we can ensure that the edges created by previous eliminations are only inside $V_o \cup \{s_{c,i} \mid c\text{ is a child of }o\text{ and }i \in [M]\}$ for each operation $o$.

  Now, we describe the details of the eliminations.
  Let $o$ be the current operation to process.
  If $o$ is the root, the number of remaining vertices is $\cw + O(\log n)$.
  Thus, the elimination of these vertices has degree $\cw + O(\log n)$.
  Otherwise, let $p$ be the parent of $o$.
  First, we eliminate vertices $Y_o$.
  Because each vertex of $Y_i$ is adjacent to at most one vertex of $Y_p$, Lemma~\ref{lem:tw:eliminate_matching} gives the elimination of degree $|N[Y_o] \setminus Y_p| \leq |V_o| \leq \cw + O(\log n)$.
  Then, we eliminate the remaining vertices $V_o \setminus Y_o$.
  If $o$ is the only child of $p$, let $V_q=Y_q=\emptyset$, and otherwise, let $q$ be the another child of $p$.
  By applying Lemma~\ref{lem:tw:eliminate_subset}, we obtain the elimination of degree $|N[V_o \setminus Y_o]|-1 \leq
  |V_o \setminus Y_o| + |V_p| + |V_q\setminus Y_q| \leq \cw + O(\log n)$.
\end{proof}

%!TEX root = ../paper.tex

\section{From 3SAT parameterized by tree-width to Independent Set parameterized by clique-width}\label{sec:tw-3sat-to-cw-is}

Let $(X,\mathcal{C} = \{C_1,\ldots,C_m\})$ be an instance of $3$-SAT with tree-width $\tw$.
We want to construct an instance $(G,k)$ of \IS with clique-width $\tw + O(\log \tw)$ such that $G$ has an independent set of size at least $k$ if and only if $\mathcal{C}$ is satisfiable.
For this purpose, we use the same construction of $(G,k)$ as in
Section~\ref{sec:tw:3sat_is}.
Hence, it suffices to show that the graph $G$ has clique-width $\tw + O(\log \tw)$.

\begin{lemma}
  The graph $G$ has clique-width at most $\tw + O(\log \tw)$.
\end{lemma}
\begin{proof}
  Let $T = (I, F)$ be a nice tree-decomposition of $(X,\mathcal{C})$.
  We inductively construct $G$ by processing each node of $T$ in a bottom-up manner.

  For a node $i \in I$, let $I_i^\downarrow \subseteq I$ be the set consisting of $i$ itself and descendants of $i$.
  Then, we define $X_i^\downarrow$ and $\mathcal{C}_i^\downarrow$ as the sets of variables and constraints, respectively, contained in a bag of $I_i^\downarrow$.
  Let $G_i^\downarrow$ be the subgraph of $G$ induced by variable gadgets corresponding to vertices in $X_i^\downarrow$, clause gadgets corresponding to clauses in $\mathcal{C}_i^\downarrow$, child counting gadgets for nodes in $I_i^\downarrow$ and parent counting gadgets for nodes in $I_i^\downarrow \setminus \{i\}$.
  At node $i$, we will construct the graph $G_i^\downarrow$.

  We introduce a special label $\#$; if a vertex is once labeled $\#$, then we will never relabel or connect new edges to that vertex.
  For each $i \in I$, we ensure that vertices $\overline{x}_i$ for $x \in X$ and vertices
  in the last layer of the child counting gadget for $i$ have distinct labels, and all the other vertices in
  $G_i^\downarrow$ are labeled $\#$ after processing the node $i$.

  Suppose that we have constructed $G_c^\downarrow$ for a child node $c$ of $i$ (if $i$ is a Join node, we also have another graph $G_{c'}^\downarrow$ for the other child $c'$), and we want to construct a graph $G_i^\downarrow$.
  We have five cases depending on the type of the node $i$.

  (i) If $i$ is a leaf node, we have nothing to do.

  (ii) Suppose $i$ is an Introduce($x$) node.
  Let $X_c = \{x^1,\ldots,x^d\}$ for some $d \leq \tw$.
  Note that $X_i = \{x^1,\ldots,x^d,x\}$ holds.

  For each $j \in [d]$, we do the following:
  We first construct a variable gadget for $x^j$ using new labels.
  We then connect $\overline{x}^j_c$ to $x^j_i$, and the label of $\overline{x}^j_c$ is set to $\#$.
  Next, we create the $j$-th layer of the child counting gadget $S_i$ for $i$, and connect $\overline{x}^j_i$ to it.
  This can be done using auxiliary $O(\log \tw)$ labels.
  Then, the labels of $(j-1)$-th layer (if exists) of $S_i$ and the label of $x^j_i$ are set to $\#$.
  Finally, we create the $j$-th layer of the parent counting gadget $T_c$ for $c$, and connect $\overline{x}^j_i$ to it.
  This can be done using auxiliary $O(\log \tw)$ labels.
  Then, the labels of $(j-1)$-th layer (if exists) of $T_c$ are set to $\#$.

  After processing $x^1,\ldots,x^d$, we create a variable gadget for $x_i$ and connect $x_i$ with the $(d+1)$-th layer of $S_i$ for $i$.
  Then, the labels of $d$-th layer (if exists) of $S_i$ are set to $\#$.
  Finally, we connect the last layers of $T_c$ and the child counting gadget $S_c$ for $c$ to make the constraint $|\{x^j_c \mid j \in [d]\} \cup \{\overline{x}^j_i \mid j \in [d]\} \cap S| \leq d$ for any independent set $S$.
  In total, we only need $\tw + O(\log \tw)$ labels.

  (iii) Suppose $i$ is an Introduce($x \vee y \vee z$) node.
  The construction is very similar to the case (ii).
  The only difference is that we have to make a clause gadget corresponding to the
  clause $(x \vee y \vee z)$, where $x$, $y$, and $z$ are literals.
  Recall that, in the case (ii), the label of $x^j_i$ is set to $\#$ after the $j$-th iteration.
  Instead, if $x^j_i$ is the literal used in the clause, then we keep it using a new label.
  After the $d$-th step, we construct a clause gadget using these kept literals.
  We only need $O(1)$ auxiliary labels for this construction since the clause $(x \vee y \vee z)$ has only three literals.

  (iv) If $i$ is a Forget($v$) node or (v) a Join node, then the construction is almost the same as (ii), and we omit the detail.

%  (v) If $i$ is a join node, then the construction is also almost the same as (ii), and we omit the detail.

  To summarize, we can construct $G$ using $\tw + O(\log \tw)$ labels.
\end{proof}

%!TEX root = ../paper.tex

\section{Exactly Parameterized $\NL$}\label{sec:pw}
%Let $G=(V,E)$ be a graph and $\pi: V\rightarrow [|V|]$ be a bijection (ordering).
%The \emph{band-width} of an ordering $\pi$ is the maximum of $|\pi(u)-\pi(v)|$ over all edges $uv\in E$,
%and the band-width of a graph $G$, $\bdw(G)$, is the minimum band-width among all possible orderings.
%For any graph $G$, the band-width $\bdw(G)$ is at least the path-width $\pw(G)$~\cite{hoge}.

By extending the classical complexity class $\NL$ (Non-deterministic Logspace), we define a class of parameterized
problems $\NLk$ (Exactly Parameterized $\NL$) which can be solved by a non-deterministic Turing machine with the space
of $k+O(\log n)$ bits.

\begin{definition}[$\NLk$]
A parameterized problem $(L, \kappa)$, where, $L\subseteq \{0,1\}^*$ is a language and $\kappa:
\{0,1\}^*\rightarrow\mathbb{N}$ is a parameterization, is in $\NLk$ if there exists a
polynomial $p:\mathbb{N}\rightarrow \mathbb{N}$ and a verifying polynomial-time deterministic Turing machine $M:
\{0,1\}^*\times \{0,1\}^*\rightarrow\{0,1\}$ with four binary tapes, a read-only input tape, a read-only read-once
certificate tape, and two read/write working tapes called the $k$-bit tape and the logspace tape with the following properties.
\begin{itemize}
  \item For any input $x \in \{0,1\}^*$, it holds that $x\in L$ if and only if there exists a certificate
  $y\in\{0,1\}^{p(|x|)}$ such that $M(x,y)=1$.
  \item For any $x\in \{0,1\}^*$ and $y\in \{0,1\}^{p(|x|)}$, the
  machine $M$ uses at most $\kappa(x)$ space from the $k$-bit tape and $O(\log |x|)$ space from the logspace tape.
\end{itemize}
\end{definition}

Note that the machine $M$ is not allowed to use $O(\kappa(x))$ bits from the $k$-bit tape but at most $\kappa(x)$
bits.
This is why we use two separated working tapes instead of one long working tape of length $\kappa(x)+O(\log |x|)$; in
the latter case, because there is only one head, it may be difficult to simulate a random-access $\kappa(x)$-bit array.

We give several examples of problems in $\NLk$.
%Due to the space constraints, the proofs are deferred to Appendix~\ref{sec:pw_omitted}.
%Due to the space constraints, the proofs are omitted.
For all the problems in Lemma~\ref{lem:pw:nlk2}, the current fastest algorithms take $O^*(2^n)$
time~\cite{DBLP:journals/mst/BodlaenderFKKT12}.

\begin{lemma}\label{lem:pw:nlk1}
\SAT, \ThreeSAT, \MaxTwoSAT, and \IS parameterized by path-width are in $\NLk$.
\end{lemma}
\begin{proof}
We show that \SAT parameterized by path-width is in $\NLk$.
%Since the band-width is always at least the path-width, this implies that \SAT parameterized by band-width is also in
%$\NLk$.
For the other problems, we can use similar proofs, so we omit them.

Let $(X_1, \ldots, X_d)$ be the list of bags of the nice path-decomposition from the root to the leaf.
As a certificate, we use a list of partial assignments $f_i:X_i\rightarrow\{0,1\}$.
Starting from the root bag $X_1$, the machine $M$ handles each bag one by one as follows.
Let $X_i$ be the current bag.
By storing the current partial assignment $f_i$ to the $k$-bit tape, we can check that there are no
inconsistencies between two assignments $f_i$ and $f_{i-1}$.
From the definition of path-decomposition, if $f_i$ and $f_{i-1}$ are consistent for all $i$, all the partial
assignments are consistent.
If $X_i$ is an Introduce($C$) bag, we check that the partial assignment satisfies the clause $C$.
Since each clause $C$ has an Introduce($C$) node, this implies that the assignment given as the certificate satisfies
all the clauses.
\end{proof}

\begin{lemma}\label{lem:pw:nlk2}
\TSP, \problem{Optimal Linear Arrangement}, \problem{Directed Feedback Arc
  Set} parameterized by the number of vertices of the input graph, and \SetCover parameterized by the number of elements 
  are in $\NLk$.
\end{lemma}
\begin{proof}
We show that \TSP parameterized by the number of vertices is in $\NLk$.
For the other problems, we can use similar proofs, so we omit them.
\TSP is the following problem: given a directed graph $G=(V,E)$ answer whether there exists a cycle that passes each
vertex exactly once.

As a certificate, we use an ordering of vertices on the cycle.
Then the machine reads each vertex in the ordering one by one.
We can check the ordering is actually a cycle by putting the first and the
last vertex on the logspace tape.
Since the certificate tape is read-once, we cannot check whether each vertex appears exactly once by only using logspace
tape.
When the machine reads a vertex $i$ from the certificate, it writes a symbol 1 on the $i$-th position of the $k$-bit
tape.
If the symbol in the $i$-th position is already 1, the certificate contains the vertex $i$ multiple times.
Finally, by checking all the symbols in the $k$-bit tape is 1, we can confirm that each vertex appears exactly once in
the ordering.
\end{proof}

Now, we define \emph{logspace parameter-preserving reduction} and introduce $\NLk$-complete problems.

\begin{definition}[Reducibility]
A parameterized problem $A=(L,\kappa)$ is \emph{logspace parameter-preserving reducible} to a parameterized problem
$B=(L',\kappa')$, denoted by $A\nlkreduce B$, if there exists a logspace computable function $\phi:\{0,1\}^*\rightarrow
\{0,1\}^*$ such that
\begin{itemize}
  \item $x\in L\iff \phi(x)\in L'$, and
  \item $\kappa'(\phi(x))\leq \kappa(x)+O(\log |x|)$.
\end{itemize}
\end{definition}

Note that in the standard parameterized reduction, the computation can take $f(\kappa(x))\poly(|x|)$ time and the
parameter $\kappa'(\phi(x))$ of the reduced instance can be increased to any function of the original parameter $\kappa(x)$.
However, in our reduction, we allow only a logspace computation and an additive increase by $O(\log |x|)$ of the
parameter.

\begin{proposition}
If $A\nlkreduce B$ and $B\in \NLk$, then $A\in\NLk$.
\end{proposition}
The proof of the proposition is an easy extension of the case for $\NL$ (see the text
book by Arora and Barak~\cite[Chap.4.3.]{DBLP:books/daglib/0023084}), so we omit it here.

\begin{definition}[$\NLk$-complete]
A parameterized problem $A$ is called \emph{$\NLk$-hard} if for any $B\in \NLk$, we have $B\nlkreduce
A$.
Moreover, if $A\in \NLk$, $A$ is called \emph{$\NLk$-complete}.
\end{definition}

Since there are at most $2^{k+O(\log |x|)}\poly(|x|)=O^*(2^k)$ states, any
problem in $\NLk$ can be solved in $O^*(2^k)$ time by dynamic programming.
The following proposition follows from the definitions.
%From the definitions, if one of the $\NLk$-hard problems can be solved in $O^*(c^k)$ time, then any $\NLk$
%problem can also be solved in $O^*(c^k)$ time.
\begin{proposition}
  Any problem in $\NLk$ can be solved in $O^*(2^k)$ time.
  If one of the $\NLk$-hard problem can be solved in $O^*(c^k)$ time, then any problem in $\NLk$ can also be
  solved in $O^*(c^k)$ time.
\end{proposition}

% Now, we show that the problems \SAT, \ThreeSAT, \MaxTwoSAT, and \IS parameterized by
% path-width are $\NLk$-complete.
Now, we show that the problems in Lemma~\ref{lem:pw:nlk1} are $\NLk$-complete.
%Due to the space constraints, we give a proof only for \SAT here.
%The proofs for the other problems are deferred to Appendix~\ref{sec:pw_omitted}.
%The proofs for the other problems can be obtained by series of
% (decomposition-based) reductions and will be presented in the full version of this paper.

%Note that since the band-width is always at least the path-width, this implies that these problems parameterized by
%path-width are also $\NLk$-complete.
%In order to bound the band-width, we use the following idea.
%We construct a graph such that a layer $K_i$ of size $k$ and a layer $L_i$ of size $O(\log |x|)$ are piled up
%alternatively ($V=L_0\cup K_1\cup L_1\cup \ldots\cup K_n\cup L_n$).
%Each $L_i$ can only be connected to $L_{i-1}\cup K_i\cup L_i\cup K_{i+1}\cup L_{i+1}$, and $j$-th vertex $k_{i,j}$
%of $K_i$ can only be connected to $\{k_{i-1,j+O(1)}\}\cup L_{i-1}\cup K_i\cup L_i\cup \{k_{i+1,j+O(1)}\}$.
%Then, by taking vertices in order, we can obtain an ordering of band-width at most $k+O(\log|x|)$.

\begin{theorem}
SAT parameterized by path-width is $\NLk$-complete.
\end{theorem}
\begin{proof}
SAT parameterized by path-width is in $\NLk$.
So it is sufficient to show that any parameterized problem $A=(L,\kappa)$ in $\NLk$ can be reduced to SAT
parameterized by path-width.
Let $M$ be a Turing machine that accepts $L$, $Q$ be the set of (internal) states
of $M$, and $t, s: \mathbb{N}\rightarrow\mathbb{N}$
be the polynomial time bound and logarithmic space bound of $M$, respectively.
We reduce an instance $x$ of $A$ with a parameter $k=\kappa(x)$ to \SAT as follows.

For each step $i\in [t(|x|)]$, we create the following variables:
\begin{itemize}
  \item $Q_{i,q}$ for each $q\in Q$, which indicates that $M$ is in state $q$,
  \item $H^I_{i,j}$ for each $j\in [\ceil{\log|x|}]$, which indicates the position of the input tape head in binary,
  \item $H^K_{i,j}$ for each $j\in [\ceil{\log k}]$, which indicates the position of the $k$-bit tape head in binary,
  \item $H^L_{i,j}$ for each $j\in [\ceil{\log r(|x|)}]$, which indicates the position of the logspace tape head in binary,
  \item $T^K_{i,h}$ for each $h\in [k]'$, which indicates the symbol written in the $h$-th cell of the $k$-bit tape,
  \item $T^L_{i,h}$ for each $h\in [s(|x|)]'$, which indicates the symbol written in the $h$-th cell of the logspace
  tape, and
  \item $T^C_i$, which represents the symbol in the cell of the certificate tape.
\end{itemize}
%Let $T^K_i=\{T^K_{i,h}\mid h\in[k]'\}$ and $X_i$ be the set of variables other than $T^K_i$.
%We create copies $X'_i=\{x'\mid x\in X_i\}$ and create a constraint $x=x'$ for each $x\in X_i$.
%These copies are created to bound the band-width.

Now, we create clauses.
%For an integer $i$, we write the set $\{j\mid (i\enspace\mathrm{AND}\enspace2^j)\neq 0\}$ by $\mathrm{bit}(i)$, where
%$\mathrm{AND}$ is the binary and operator.
Let $q_s\in Q$ be the initial state and $q_t\in Q$ be the accepting state.
First, we create the following clauses (consisting of single literals) to express the initial and the final configuration:
\begin{itemize}
  \item $Q_{1,q_s}$ (the machine is in the state $q_s$),
  \item $\overline{H^I_{1,j}}$ for each $j\in [\ceil{\log|x|}]$ (the input tape head is at the position 0),
  \item $\overline{H^K_{1,j}}$ for each $j\in [\ceil{\log k}]$ (the $k$-bit tape head is at the position 0),
  \item $\overline{H^L_{1,j}}$ for each $j\in [\ceil{\log s(|x|)}]$ (the logspace tape head is at the position 0),
  \item $\overline{T^K_{1,h}}$ for each $h\in [k]'$ (each cell of the $k$-bit tape has symbol 0),
  \item $\overline{T^L_{1,h}}$ for each $h\in [r(|x|)]'$ (each cell of the logspace tape has symbol 0), and
  \item $Q_{t(|x|),q_t}$ (the machine must finish in the state $q_t$).
\end{itemize}

Then, for each step $i\in [t(|x|)]$, we create clauses to express transitions.
The machine can take only one state, so we create a clause $\overline{Q_{i,q}}\vee\overline{Q_{i,q'}}$ for each $q\neq
q'$.
If a cell changes, the head must be there (or equivalently, cells not pointed by the head must remain unchanged), so
we create the following clauses:
\begin{itemize}
  \item $T^K_{i,h^K}\neq T^K_{i+1,h^K}\rightarrow (H^K_{i,*})_2=h^K$ for each $h^K\in [k]'$, and
  \item $T^L_{i,h^L}\neq T^L_{i+1,h^L}\rightarrow (H^L_{i,*})_2=h^L$ for each $h^L\in [s(|x|)]'$.
\end{itemize}
Let $\delta: (q,c^I,c^K,c^L,c^C)\mapsto(q',c'^K,c'^L,d^I,d^K,d^L,d^C)$ be the transition function, which indicates
that if the machine is in the state $q$, the symbol in the input tape is $c^I$, the symbol in the $k$-bit tape is $c^K$,
the symbol in the logspace tape is $c^L$, and the symbol in the certificate tape is $c^C$, then the machine changes the
state to $q'$, write $c'^K$ to the cell of the $k$-bit tape, write $c'^L$ to the cell of the logspace tape, move the
input tape head by $d^I$, move the $k$-bit tape head by $d^K$, move the logspace tape head by $d^L$, and move the
certificate tape head by $d^C$.
Note that since the certificate tape is read-once, $d^C\geq 0$.
For each $h^I\in [|x|]'$, $h^K\in [k]'$, $h^L\in [s(|x|)]'$, and transition
$(q,c^I,c^K,c^L,c^C)\mapsto(q',c'^K,c'^L,d^I,d^K,d^L,d^C)$, we create clauses as follows.
If a symbol in the $h^I$-th position of the input tape is not $c^I$, this transition never occurs.
Otherwise, let $C$ be the constraint $Q_{i,q}\; \wedge\; (H^I_{i,*})_2=h^I\;\wedge\;
(H^K_{i,*})_2=h^K\;\wedge\; (H^L_{i,*})_2=h^L\;\wedge\; T^K_{i,h^K}=c^K\;\wedge\; T^L_{i,h^L}=c^L\;\wedge\; T^C_i=c^C$.
Then, we create the following clauses:
\begin{itemize}
  \item $C\rightarrow Q_{i+1,q'}$ (the machine changes the state to $q'$),
  \item $C\rightarrow T^K_{i+1,h^K}=c'^K$ ($c'^K$ is written in the cell of the $k$-bit tape),
  \item $C\rightarrow T^L_{i+1,h^L}=c'^L$ ($c'^L$ is written in the cell of the the logspace tape),
  \item $C\rightarrow (H^I_{i+1,*})_2=h^I+d^I$ (the input tape head moves by $d^I$),
  \item $C\rightarrow (H^K_{i+1,*})_2=h^K+d^K$ (the $k$-bit tape head moves by $d^K$),
  \item $C\rightarrow (H^L_{i+1,*})_2=h^L+d^L$ (the logspace tape head moves by $d^L$), and
  \item $C\rightarrow T^C_i=T^C_{i+1}$ if $d^C=0$ (if the certificate tape head does not move, then the symbol in the certificate tape does not change).
\end{itemize}

It is not difficult to check that the reduction can be done in logspace and the obtained CNF is satisfiable if and only
if there is a certificate such that the machine finishes in the accepting state.
Finally, we show that the obtained CNF has path-width $k+O(\log |x|)$.

For a step $i$, let $T^K_i=\{T^K_{i,h}\mid h\in[k]'\}$ and $X_i$ be the set of other variables.
The primal graph of the obtained CNF has the following properties:
\begin{itemize}
  \item $N[X_i]\subseteq T^K_{i-1}\cup X_{i-1}\cup T^K_i\cup X_i\cup T^K_{i+1}\cup X_{i+1}$,
  \item $N(T^K_{i,j})\subseteq \{T^K_{i-1,j},T^K_{i+1,j}\}\cup X_{i-1}\cup X_i\cup X_{i+1}$.
\end{itemize}
We can construct a path-decomposition as follows:
starting from a bag $T^K_1\cup X_1$ and $i=1$, introduce $X_{i+1}$, introduce $T^K_{i+1,1}$, forget $T^K_{i,1}$,
\ldots, introduce $T^K_{i+1,k}$, forget $T^K_{i,k}$, forget $X_i$ (the current bag consists of $T^K_{i+1}\cup X_{i+1}$),
and then increase $i$.
Since the size of $X_i$ is $O(\log |x|)$ and the size of $T^K_i$ is exactly $k$, the width of the obtained
path-decomposition is $k+O(\log |x|)$.
\end{proof}

\begin{theorem}
\ThreeSAT parameterized by path-width is $\NLk$-complete.
\end{theorem}
\begin{proof}
We prove the theorem by a reduction from \SAT parameterized by path-width.
The reduction is completely the same as the standard reduction (see Section~\ref{sec:tw:sat_3sat}).
Starting from an empty bag and the leaf node $i$ of the given nice path-decomposition of width $\pw$, we can construct a
path-decomposition of the reduced instance as follows.
If $i$ is an Introduce($C$) node of length more than three, let $\{y_1,\ldots,y_k\}$ be the variables
created to replace the clause $C$.
Then, we introduce $y_1$, introduce $y_2$, forget $y_1$, introduce $y_3$, forget $y_2$, \ldots, introduce
$y_k$, forget $y_{k-1}$, and forget $y_k$.
If $i$ is an Introduce($x$) node, we introduce $x$, and if $i$ is a Forget($x$) node, we forget $x$.
Finally, we change $i$ to its parent and repeat the process until reaching to the root.
The width of this path-decomposition is $\pw+O(1)$.
\end{proof}

\begin{theorem}
\IS parameterized by path-width is $\NLk$-complete.
\end{theorem}
\begin{proof}
We prove the theorem by a reduction from \ThreeSAT parameterized by path-width.
The reduction is completely the same as that for the tree-width case (Section~\ref{sec:tw:3sat_is}), so we only need to
bound the path-width of the obtained graph.
Starting from an empty bag and the leaf node $i$ of the given nice path-decomposition of width $\pw$, we can construct a
path-decomposition of the reduced instance as follows.

If $i$ is not the leaf, let $c$ be the child of $i$.
For each variable $x\in X_i\cap X_c$, we introduce $x_i$ and forget $\overline{x_c}$.
If $i$ is an Introduce($x$) node, we introduce $x_i$, if $i$ is a Forget($x$) node, we forget $\overline{x_c}$, and if
$i$ is an Introduce($C$) node, we introduce the corresponding clause gadget.

Then, we process the child counting gadget for $i$ as follows.
First, we introduce the first layer $S_{i,1}$.
Then, starting from $a=1$, we repeat the following process by incrementing $a$:
(1) introduce the next layer $S_{i,a+1}$, (2) for each clause gadget $C$ connecting $S_{i,a}$ and $S_{i,a+1}$, introduce
$C$ and forget $C$ one by one, (3) forget the current layer $S_{i,a}$.
Note that the last layer of the counting gadget is remained in the bag.

Next, for each variable $x\in X_i$, we introduce $\overline{x_i}$ and forget $x_i$ one by one.
If $i$ is an Introduce($C$) node, we forget the corresponding clause gadget.
We process the parent counting gadget for $c$ in the same way as we did for the child counting gadget.
Then, we process the last layers of the child and the parent counting gadget for $c$.
For each clause gadget $C$ connecting the last layers of the child and the parent counting gadget, we introduce
$C$ and forget $C$ one by one, and then we forget these two layers .
Finally, we change $i$ to its parent and repeat the process until reaching to the root.
The width of this path-decomposition is $\pw+O(\log \pw)$.
\end{proof}

\begin{theorem}
\MaxTwoSAT parameterized by path-width is $\NLk$-complete.
\end{theorem}
\begin{proof}
We prove the theorem by a reduction from \IS parameterized by path-width.
The proof is completely the same as that for the tree-width case (Section~\ref{sec:tw:is_max2sat})
\end{proof}

\subsubsection*{Acknowledgement}
Yoichi Iwata is supported by JSPS Grant-in-Aid for JSPS Fellows (256487).
Yuichi Yoshida is supported by JSPS Grant-in-Aid for Young Scientists (B)
(No.~26730009), MEXT Grant-in-Aid for Scientific Research on Innovative Areas
(24106001), and JST, ERATO, Kawarabayashi Large Graph Project.

\bibliographystyle{abbrv}
\bibliography{paper}

\end{document}